\documentclass[12pt,a4paper]{article}

\usepackage[T1]{fontenc}
\usepackage[utf8]{inputenc}
\usepackage{lmodern}
\usepackage{microtype}
\usepackage{geometry}
\usepackage{amsmath,amssymb,amsthm}
\usepackage{booktabs}
\usepackage{tabularx}
\usepackage{array}
\usepackage{hyperref}
\hypersetup{
    colorlinks=true,
    linkcolor=blue,
    citecolor=blue,
    urlcolor=blue
}
\usepackage{url}
\usepackage{enumitem}
\usepackage{titlesec}
\usepackage{abstract}
\usepackage{parskip}
\usepackage{xcolor}

\theoremstyle{plain}
\newtheorem{theorem}{Theorem}[section]
\newtheorem{corollary}[theorem]{Corollary}

\theoremstyle{definition}
\newtheorem{definition}{Definition}[section]
\newtheorem{assumption}{Assumption}[section]

\theoremstyle{remark}
\newtheorem{remark}{Note}[section]

\title{\textbf{Broadcast Zero-Trust Edge Computing: Formal Threat Model
Reduction for Air-Gapped AI Terminals via Physically
Unidirectional Broadcast Datacasting}}

\author{
    Vasu Srinivasan\\
    \small 6Hats AI Labs Inc.\\
    \small \href{mailto:vasu@6hats.ai}{vasu@6hats.ai}
    \and
    Dhriti Vasu\\
    \small University of California, Berkeley\\
    \small \href{mailto:dhritivasu@berkeley.edu}{dhritivasu@berkeley.edu}
}

\date{
\small A provisional patent application covering aspects of this
architecture has been filed.}

\begin{document}
% -----------------------------------------------------------------------

\maketitle

% -----------------------------------------------------------------------
\begin{abstract}
We present a formal argument under a defined adversary model establishing
threat model reduction for a Sovereign AI Architecture---an air-gapped
intelligent terminal performing all inference on-device with receive-only
inbound data is delivered via a physically unidirectional channel—implemented 
using receive-only broadcast infrastructure or certified hardware 
data diodes—with no return path to any external network, and session key 
material is delivered via an out-of-band optical channel. We demonstrate
across four complementary frameworks that this architecture eliminates all
network-mediated compromise vectors reachable by a remote adversary by
construction: attack surface decomposition, remote adversary channel
reachability, lateral movement graph isolation, and cryptographic session
architecture. Under a precisely stated remote adversary model, we establish
that the set of network-reachable attack channels available to a remote
adversary is empty by construction.

We further establish that a physically unidirectional inbound channel is a
pivotal enabling primitive for this architecture. The property required is
that the receiver does not transmit on the inbound channel during normal
operation. This property is enforced at different assurance levels depending
on deployment configuration: by hardware absence of transmit circuitry in
receive-only broadcast and satellite terminals, by certified hardware in data
diode deployments, and by software policy in managed IP deployments. The
argument is transport-agnostic and generalizes to any inbound channel
satisfying Definition~\ref{def:unidirectionality}. Concrete implementations
across these configurations are described in Appendix~A.

We show that self-propagating network worms, botnet enrollment, and all
lateral movement attack classes are impossible against a terminal that is not
a vertex in the institutional network graph. We introduce the framing of
\emph{Broadcast Zero-Trust Edge Computing}---a model in which compute is
sovereign, data ingress is physically unidirectional broadcast, identity is
optically gated, and network attack surface is eliminated by
construction---as a generalizable architecture pattern for regulated AI
deployment. This concept reframes zero-trust principles at the architecture
level: rather than authenticating every network transaction, the architecture
removes network transaction capability entirely.

A reference prototype implementation is described validating the architectural
claims. We additionally identify and correct several common overclaim patterns
found in practitioner-authored security architecture proposals. The residual
attack surface---bounded to line-of-sight physical presence---is the
irreducible minimum for any deployed physical system.

\end{abstract}

\tableofcontents
\newpage

% -----------------------------------------------------------------------
\section{Introduction}
% -----------------------------------------------------------------------

The deployment of AI inference systems in highly regulated environments---clinical
care, financial transaction processing, defense intelligence---creates a fundamental
tension between connectivity and sovereignty. Networked AI terminals inherit the
attack surface of the networks they join. In healthcare, hospital networks are among
the most frequently breached infrastructure categories in the United
States~\cite{hipaajournal2024,ponemon2023}. In financial services,
network-connected terminals are the primary target of advanced persistent threat
(APT) actors~\cite{verizonDBIR2023}. In defense, lateral movement through
enterprise networks has been demonstrated as a viable pathway to sensitive systems
even behind classified boundaries~\cite{cisaAPT2022}.

The conventional response to this tension is hardening: firewalls, VLANs,
intrusion detection, endpoint protection, and network access control. These controls
reduce the probability of exploitation but do not eliminate the attack class. They
are probabilistic guarantees subject to misconfiguration, zero-day vulnerabilities,
insider threats, and supply chain compromise.

This paper examines a structurally different approach: a \emph{Sovereign AI
Architecture} in which the terminal has no network interface, all AI inference runs
entirely on-device, inbound data is delivered via a physically unidirectional 
channel—implemented using receive-only broadcast infrastructure or certified hardware 
data diodes—with no return path to any external network, and session key material 
is delivered via an out-of-band optical channel. We demonstrate that this architecture 
eliminates the network-mediated attack class entirely---not by hardening, but by removal.

A central technical contribution of this paper is the formal characterization of
receiver-side unidirectionality across terrestrial, satellite, and radio broadcast
classes---as a pivotal enabling technology for sovereign architectures. The
architecture requires one property from the transport: deliver signed encrypted
packets in one direction. Three deployment configurations satisfy this requirement:
a hardware data diode on a managed wire, a managed IP connection with outbound
capability disabled, or a physically unidirectional broadcast channel. The formal
argument is independent of which configuration is chosen.

The security guarantee lives in the packet verification step---signed encrypted
packets verified before any processing, output confined to an optical
channel---not in the transport. Deployment configurations differ in how strongly
they enforce unidirectionality: hardware configurations enforce it independent of
software correctness, policy configurations enforce it subject to correct
implementation and operation. Appendix~A describes three concrete implementations
across these configurations for reference.

We additionally introduce the concept of \emph{Broadcast Zero-Trust Edge Computing}
as a generalizable architecture pattern. This concept reframes zero-trust principles
at the architecture level: rather than authenticating every network transaction, the
architecture removes network transaction capability entirely. The pattern has four
defining properties: sovereign compute, physically unidirectional broadcast ingress,
optically gated session identity, and complete elimination of the remote attack
surface.

We further prove that three of the most consequential real-world attack classes
-- self-propagating network worms (WannaCry~\cite{wannacry}, Slammer~\cite{slammer}),
botnet enrollment (Mirai, TrickBot), and lateral movement (SolarWinds aftermath) 
-- are impossible against a terminal that is not a vertex in the institutional network
graph. These results follow directly from the graph isolation argument and require no
additional assumptions.

The paper is structured as follows. Section~\ref{sec:threatmodel} presents threat
model definitions and scope. Section~\ref{sec:attacksurface} formalizes attack
surface decomposition. Section~\ref{sec:broadcast} establishes receiver-side
unidirectionality as the sovereign inbound channel property and surveys deployment
configurations satisfying it. Section~\ref{sec:reachability} presents the remote
adversary channel reachability argument. Section~\ref{sec:lateral} presents lateral
movement graph isolation, worm propagation impossibility, and botnet enrollment
impossibility. Section~\ref{sec:crypto} presents the cryptographic session and key
delivery architecture. Section~\ref{sec:composite} presents the composite threat
model reduction. Section~\ref{sec:discussion} discusses implementation, limitations,
the Broadcast Zero-Trust framing, and fleet health monitoring via session lifecycle
events. Section~\ref{sec:conclusion} concludes.

Throughout, we distinguish between information-theoretic security (holding against
computationally unbounded adversaries) and computational security (holding against
polynomial-time adversaries). The session cryptography employed provides
computational security. The channel reachability and graph isolation arguments are
structural---they hold regardless of adversary compute.

% -----------------------------------------------------------------------
\section{Threat Model and Scope}
\label{sec:threatmodel}
% -----------------------------------------------------------------------

\subsection{Adversary Model}

We consider two adversary classes:

\begin{definition}[Remote Adversary $\mathcal{A}_r$]
A remote adversary is a computationally bounded adversary with no physical presence
at the terminal location. $\mathcal{A}_r$ may:
\begin{itemize}
  \item Control arbitrary nodes on any network connected to the institution
  \item Intercept, modify, or inject traffic on any network path
  \item Operate compromised vendor or supply chain network nodes
  \item Possess nation-state-level network capabilities
\end{itemize}
$\mathcal{A}_r$ may \emph{not}:
\begin{itemize}
  \item Physically access the terminal or its immediate environment
  \item Present optical signals to terminal cameras or displays
  \item Present USB devices to terminal ports
  \item Operate RF transmission equipment in physical proximity to the terminal
    antenna
\end{itemize}
\end{definition}

\begin{definition}[Physical Adversary $\mathcal{A}_p$]
A physical adversary is an adversary with line-of-sight or direct physical access
to the terminal. $\mathcal{A}_p$ may perform any action that $\mathcal{A}_r$ may
perform, and additionally may:
\begin{itemize}
  \item Present optical signals within line-of-sight of the terminal
  \item Physically interact with terminal ports and peripherals
  \item Operate RF transmission equipment in proximity to the terminal
  \item Attempt chassis intrusion
\end{itemize}
\end{definition}

This paper focuses primarily on eliminating the $\mathcal{A}_r$ attack class.
Residual $\mathcal{A}_p$ risks are discussed in Section~\ref{sec:residual}.

\subsection{Scope and Assumptions}

The following assumptions bound the argument. We state them explicitly to avoid the
overclaim failure mode identified in the security literature~\cite{shostack2014,saltzer1975}.

\begin{assumption}[Hardware Trust]
\label{asm:hardware}
We assume the terminal hardware, secure enclave, and firmware were not compromised
during manufacture or supply chain. Hardware trust is a prerequisite for any
deployed security system. Supply chain risks are discussed in
Section~\ref{sec:limitations}.
\end{assumption}

\begin{assumption}[Cryptographic Primitives]
\label{asm:crypto}
We assume AES-256-GCM and ECC (NIST P-256 or Curve25519) are computationally secure
against polynomial-time adversaries. Both are standardized under FIPS
140-3~\cite{fips1403} and NIST SP 800-186~\cite{nistSP800186}. We make no
information-theoretic claims about these primitives.
\end{assumption}

\begin{assumption}[Authenticated Inbound Channel]
\label{asm:auth}
We assume all inbound data received via the broadcast channel is authenticated via a
cryptographic signature verified prior to any processing. Unauthenticated data is
rejected before parsing. This assumption is critical---without it, malicious input
attacks via the broadcast channel remain possible (see Section~\ref{sec:malicious}).
\end{assumption}

\begin{assumption}[Physical Security Perimeter]
\label{asm:physical}
We assume the terminal is deployed within a physical security perimeter consistent
with its regulatory environment (e.g., clinical facility access controls, financial
facility physical security, classified facility physical protection).
\end{assumption}

\subsection{What This Paper Does Not Claim}

Following best practice in high-assurance system papers~\cite{sel4,intelsgx,titansecurity},
we explicitly state what this paper does \emph{not} prove:

\begin{itemize}
  \item We do not claim the system is secure against $\mathcal{A}_p$ adversaries
    with unlimited physical access
  \item We do not claim information-theoretic security for the cryptographic session
  \item We do not claim hardware or firmware integrity (Assumption~\ref{asm:hardware})
  \item We do not claim the system is free of implementation vulnerabilities in
    software components
  \item We do not claim the broadcast transmitter infrastructure is trustworthy
    without authentication (Assumption~\ref{asm:auth})
  \item We do not claim VID/PID USB whitelisting alone is sufficient against a
    determined physical adversary with device-spoofing capability
\end{itemize}

% -----------------------------------------------------------------------
\section{Attack Surface Decomposition}
\label{sec:attacksurface}
% -----------------------------------------------------------------------

\subsection{Structured Attack Surface Model}

Prior work in attack surface measurement~\cite{howard2003,manadhata2011} establishes
that attack surface cannot be treated as a scalar quantity. A single high-severity
interface may present greater risk than many low-severity ones. We therefore adopt a
structured decomposition rather than a cardinality argument.

\begin{definition}[Attack Surface]
The attack surface $S$ of a system is the set of interfaces through which an
adversary can interact with, inject input into, or extract output from the system.
Each interface $i \in S$ is characterized by a triple $(c, r, e)$ where:
\begin{itemize}
  \item $c$ is the channel type (network, optical, physical, acoustic)
  \item $r$ is the reachability class (remote, line-of-sight, contact)
  \item $e$ is the exploitability class (unauthenticated, authenticated,
    hardware-gated, physically-gated)
\end{itemize}
\end{definition}

\begin{definition}[Remote Attack Surface]
The remote attack surface $S_r \subseteq S$ is the subset of interfaces where
$r = \text{remote}$. This is the set reachable by $\mathcal{A}_r$.
\end{definition}

\begin{definition}[Attack Surface Decomposition]
For any system, $S$ can be decomposed as:
\[
  S = S_r \cup S_{\text{los}} \cup S_{\text{contact}}
\]
where $S_{\text{los}}$ is the line-of-sight attack surface and $S_{\text{contact}}$
is the contact (direct physical) attack surface.
\end{definition}

The absence of bidirectional network connectivity eliminates the remote network
attack surface but does not preclude physical or human-mediated compromise. Such
vectors---including social engineering, malicious insiders, and hardware
tampering---fall under the physical adversary ($\mathcal{A}_p$) considered in this
model.

\subsection{Internet-Connected Terminal Attack Surface}

For an internet-connected terminal in a regulated environment:
\[
  S_{\text{connected}} = S_{\text{network}} \cup S_{\text{physical}}
\]
where:
\begin{align*}
  S_{\text{network}} &= \{\text{NIC, LAN interface, Wi-Fi, Bluetooth,
    remote management interface}\}\\
  S_{\text{physical}} &= \{\text{USB (unenumerated), optical I/O, audio I/O,
    display, chassis}\}
\end{align*}
All elements of $S_{\text{network}}$ have $r = \text{remote}$ and vary in
exploitability class depending on configuration.

\subsection{Sovereign AI Architecture Attack Surface}

By the properties of a Sovereign AI Architecture:
\[
  S_{\text{sovereign}} = S_{\text{inbound}} \cup S_{\text{physical}}
\]
where:
\begin{align*}
  S_{\text{inbound}} &= \{\text{broadcast receiver (receive-only, authenticated,}\\
  &\quad r = \text{physical RF transmission capability in proximity to the receiver})\}\\
  S_{\text{physical}} &= \{\text{USB (hardware-whitelisted), optical I/O, audio I/O,
    display, chassis}\}
\end{align*}

Critically: $S_{\text{network}} = \emptyset$ for a Sovereign AI Architecture.
No element of $S_{\text{sovereign}}$ has $r = \text{remote}$. Therefore
$S_{\text{sovereign},r} = \emptyset$.

\begin{theorem}[Remote Attack Surface Elimination (Architectural)]
\label{thm:surfaceelim}
Under the Sovereign AI Architecture, the remote attack surface reachable by a
remote adversary $\mathcal{A}_r$ is empty.
\end{theorem}

\begin{proof}
By construction, no interface in $S_{\text{sovereign}}$ has $r = \text{remote}$.
The broadcast receiver has $r = \text{physical RF transmission capability in
proximity to the receiver}$---requiring an adversary to operate RF transmission
infrastructure at the deployment location---rather than $r = \text{remote}$
(network-reachable). USB interfaces require $r = \text{contact}$. Optical and audio
interfaces require $r = \text{line-of-sight}$. Therefore
$S_{\text{sovereign},r} = \emptyset$ under the remote adversary model of
Definition~2.1.
\end{proof}

\begin{remark}[Attack Surface Cardinality]
We do not claim that $|S_{\text{sovereign}}| < |S_{\text{connected}}|$ establishes
a security ordering. We claim the stronger structural result:
$S_{\text{sovereign},r} = \emptyset$. Cardinality comparisons are neither necessary
nor sufficient for this conclusion. Attack surface cardinality provides a lower bound
on potential exploit vectors but does not fully determine
exploitability~\cite{howard2003}.
\end{remark}

\subsection{Malicious Input via Broadcast Channel}
\label{sec:malicious}

A critical attack class not eliminated by air-gapping alone is malicious input
injection via the authenticated inbound channel~\cite{checkoway2011}. If an
adversary can influence the content of broadcast transmissions, they may attempt:

\begin{itemize}
  \item Buffer overflow in the broadcast packet decoder
  \item Parser vulnerabilities in payload deserialization
  \item Model weight poisoning via firmware update injection
  \item Compression bombs or resource exhaustion attacks
  \item Inference engine deserialization attacks
\end{itemize}

\noindent\textbf{Mitigation:} Assumption~\ref{asm:auth} requires cryptographic
signature verification prior to any payload processing. An adversary cannot inject
authenticated payloads without control of the signing key. The broadcast transmitter
operator is therefore a trusted party in this architecture, and its compromise is
treated as a supply chain threat (Section~\ref{sec:limitations}) rather than a
remote network threat.

% -----------------------------------------------------------------------
\section{Broadcast Datacasting as Sovereign Inbound Channel}
\label{sec:broadcast}
% -----------------------------------------------------------------------

\subsection{General Broadcast Datacasting Model}

The sovereign inbound channel requires one property: physical unidirectionality at
the receiver. Any broadcast medium satisfying this property is a valid
implementation. Three classes of broadcast infrastructure satisfy it:

Terrestrial broadcast. DVB-T2 (Europe, Asia, Africa), ISDB-T (Japan, South America), ATSC 3.0 (North America, South Korea), and regional standards. All operate on licensed spectrum with receive-only consumer hardware containing no transmit circuitry. ATSC 3.0 (NextGen TV) is the current North American standard, with deployments in over 60 markets as of 2024 [39]. It delivers IP-encapsulated payloads via the ROUTE protocol (A/331) [17] over a physical layer defined in A/322 [5] and the system standard A/300 [6], and supports addressed file delivery to specific terminal Device IDs within a broadcast footprint without any return path on the receiver device.

Inbound data is delivered via a physically unidirectional channel—implemented using receive-only broadcast infrastructure or certified hardware data diodes—with no return path to any external network.

\paragraph{Satellite broadcast.} DVB-S2 and DVB-S2X support IP datacasting via
Generic Stream Encapsulation over geostationary and LEO satellite links with global
coverage. Receive-only VSAT downlink terminals contain no transmit circuitry on the
downlink frequency. Satellite datacasting extends sovereign deployment to maritime,
remote, and international locations outside terrestrial coverage.

\paragraph{Radio datacasting.} HD Radio (IBOC), DAB+, and DRM support in-band
datacasting over licensed radio spectrum. Throughput is lower than satellite or
terrestrial broadcast but sufficient for session payload and firmware delivery in
bandwidth-constrained deployments.

The security argument applies to all three classes and to managed IP configurations
satisfying Definition~\ref{def:unidirectionality}. The formal properties
required---receiver-side unidirectionality, authentication of received
content---are properties of the deployment configuration, not of any specific
standard. Table~\ref{tab:configs} maps configurations to enforcement tier.

\subsection{Addressed Payload Delivery over Broadcast}

The sovereign inbound channel requires addressed payload delivery: encrypted
payloads must be routable to a specific terminal Device ID within a broadcast
footprint. This property enables per-terminal session payloads and fleet-wide
firmware updates to coexist on the same broadcast channel without requiring a
return path.

Satellite broadcast achieves this via DVB-S2X Generic Stream Encapsulation with per-stream addressing. Terrestrial broadcast achieves this via ROUTE-class protocols: DVB-T2 and ATSC 3.0 both implement the ROUTE protocol (A/331) [17] for IP-layer addressed delivery, enabling per-terminal session payloads and fleet-wide firmware distribution to coexist on a single broadcast multiplex. ATSC 3.0 additionally supports the ATSC Link-layer Protocol (ALP) for efficient IP encapsulation over the physical layer [5].
A single satellite footprint or terrestrial tower covers geographic areas containing potentially thousands of
deployed terminals, enabling deterministic fleet-wide updates with zero network infrastructure dependency at the receiving end~\cite{pearl2023}.

The architecture requires only that the channel deliver signed encrypted packets to
addressed receivers in one direction. Which broadcast medium or managed IP
configuration provides this is a deployment engineering decision.

\subsection{Unidirectionality on the Receiver Device}

The property that makes any broadcast or managed inbound channel suitable for
sovereign architecture is receiver-side unidirectionality as defined in
Definition~\ref{def:unidirectionality}. The enforcement mechanism varies by
deployment configuration; the security argument is the same in each case.

\begin{definition}[Unidirectionality (Receiver-Side)]
\label{def:unidirectionality}
A channel satisfies receiver-side unidirectionality if the receiver does not transmit
on the inbound channel during operation. This property is a characteristic of the
deployment configuration, not of the transport medium universally. Enforcement
mechanisms range from hardware absence of transmit circuitry (strongest: cannot be
subverted by software means) to certified hardware data diodes
(hardware-enforced at the boundary) to software policy on a managed interface
(holds under correct operation). The security argument of this paper holds under any
configuration satisfying this definition. The strength of the enforcement mechanism
is a deployment decision addressed in Appendix~A.
\end{definition}

Hardware-enforced configurations satisfy this definition unconditionally:
receive-only broadcast and satellite terminals contain no transmit circuitry on the
inbound frequency. Software cannot provision hardware that does not exist.
Policy-enforced configurations satisfy this definition conditionally: the return
channel is absent under correct operation but the hardware retains transmit
capability. The distinction is reflected in Table~\ref{tab:configs}.

\begin{remark}[Physical RF Injection]
An adversary could in principle attempt RF signal injection using SDR equipment or a
rogue transmitter targeting the receiver antenna. While high-gain antennas or
broadcast infrastructure can transmit from significant distances, such attacks still
require the deployment or control of RF transmission equipment operating in the
relevant broadcast spectrum.

In this threat model, $\mathcal{A}_r$ is restricted to adversaries acting through
conventional network infrastructure. Adversaries capable of transmitting RF signals
into the broadcast band---whether locally or from longer range---therefore fall
within the $\mathcal{A}_p$ threat class because they must operate specialized RF
transmission infrastructure rather than network-reachable software interfaces.
Additionally, broadcast payload authentication (Assumption~\ref{asm:auth}) ensures
that unauthenticated datacast content is rejected by the terminal, so RF signal
injection alone does not provide a viable remote compromise vector.
\end{remark}

\begin{theorem}[Return Channel Absence Under Hardware-Enforced Configuration]
\label{thm:returnchannel}
A receive-only broadcast terminal (satellite VSAT downlink, terrestrial DVB-T2, or
equivalent) cannot transmit data on the inbound channel by any software means
executing on the receiver device.
\end{theorem}

\begin{proof}
Transmission on a licensed broadcast downlink requires transmit-capable hardware on
the relevant frequency. A receive-only terminal contains no such hardware. Software
cannot provision hardware that does not exist. Therefore no software-mediated return
channel exists on the receiver device under this configuration.

\emph{Note:} Under policy-enforced configurations (managed IP with outbound
disabled, 5G Broadcast in broadcast-only device profile), the return channel is
absent under correct operation but the enforcement is a software guarantee rather
than a hardware one. The security argument of Section~\ref{sec:reachability} holds
under either enforcement tier; the distinction affects the assumptions under which
it holds, not the structure of the argument.
\end{proof}

\subsection{Comparison of Unidirectionality Enforcement Mechanisms}

Several network security architectures employ software-defined data diodes or
one-way gateways to enforce unidirectionality~\cite{tuptuk2018,waterfall2022}.
These approaches provide strong security properties but differ from physical
unidirectionality in a critical respect: their guarantees depend on the correctness
of the software implementation and hardware configuration.

Receiver-side unidirectionality (Definition~\ref{def:unidirectionality}) is a
property of the deployment configuration, not a universal property of any transport
medium. The same physical medium can be deployed in configurations that satisfy
Definition~\ref{def:unidirectionality} to different assurance levels.
Table~\ref{tab:configs} maps enforcement mechanism to assurance level.
Hardware-enforced configurations provide the guarantee independent of software
correctness. Policy-enforced configurations provide it subject to correct
implementation and operation.

\begin{table}[ht]
\centering
\caption{Inbound Channel Configurations and Unidirectionality Enforcement}
\label{tab:configs}
\begin{tabularx}{\textwidth}{@{}Xll@{}}
\toprule
\textbf{Configuration} & \textbf{Enforcement} & \textbf{Def.~\ref{def:unidirectionality}} \\
\midrule
Firewall rule & Software & No \\
Managed IP, outbound disabled & Kernel policy & Conditional \\
5G Broadcast (NR MBS) & Config, bidir platform & Conditional \\
Software data diode & Software/firmware & Conditional \\
Hardware data diode & Hardware & Yes \\
Satellite (DVB-S2/S2X) & Hardware + spectrum & Yes \\
Terrestrial (DVB-T2, ISDB-T) & Hardware + spectrum & Yes \\
Terrestrial (ATSC 3.0 / NextGen TV) & Hardware + spectrum & Yes \\
Radio (DAB+, DRM, HD Radio) & Hardware + spectrum & Yes \\
\bottomrule
\end{tabularx}
\smallskip
\begin{minipage}{\textwidth}
\small\textit{Conditional:} absent under correct operation; hardware retains transmit capability.\\
\textit{Yes:} cannot be created by software means on the receiver device.\\
Argument of Sections~\ref{sec:reachability}--\ref{sec:lateral} holds for all \textit{Yes} and \textit{Conditional} rows.
\end{minipage}
\end{table}

% -----------------------------------------------------------------------
\section{Remote Adversary Channel Reachability}
\label{sec:reachability}
% -----------------------------------------------------------------------

\subsection{Channel Reachability Under Remote Adversary Model}

\begin{definition}[Remote Compromise Condition]
\label{def:remotecompromise}
Remote compromise by $\mathcal{A}_r$ requires the existence of a channel
$C \in S$ such that:
\begin{enumerate}[label=(\roman*)]
  \item $C$ is writable by $\mathcal{A}_r$ (adversary can inject content)
  \item $C$ is readable by the terminal (terminal processes injected content)
  \item $C$ has $r = \text{remote}$ ($\mathcal{A}_r$ can reach $C$ without physical
    presence or RF transmission infrastructure)
\end{enumerate}
All three conditions must hold simultaneously.
\end{definition}

\begin{theorem}[Remote Compromise Impossibility Under Remote Adversary Model]
\label{thm:remoteimpossibility}
Under Assumption~\ref{asm:hardware} (hardware trust), Assumption~\ref{asm:auth}
(authenticated inbound channel), and the remote adversary model of Definition~2.1,
no channel $C$ satisfying conditions (i)--(iii) simultaneously exists in a Sovereign
AI Architecture.
\end{theorem}

\begin{proof}
We examine each channel in $S_{\text{sovereign}}$:

\paragraph{Channel 5.1 --- Inbound Broadcast/Managed Channel.}
By Theorem~\ref{thm:returnchannel}, under hardware-enforced configurations the
inbound channel has no return path on the receiver device. Under policy-enforced
configurations the return channel is absent under correct operation. In either case,
an adversary wishing to inject content must operate RF transmission infrastructure
within physical proximity of the receiver, or compromise the managed IP outbound
policy on the terminal---the former requires physical presence; the latter requires
either physical access or supply chain compromise of the terminal software stack.
Condition~(iii) is not satisfied for $\mathcal{A}_r$ under either configuration.

Furthermore, by Assumption~\ref{asm:auth}, all received content must pass
cryptographic signature verification before processing. An adversary without the
signing key cannot produce authenticated payloads. Condition~(i) is not satisfied
without physical compromise of the signing infrastructure.

\paragraph{Channel 5.2 --- USB Ports.}
USB enumeration is hardware-gated via kernel-level enforcement. Only devices whose
Vendor ID, Product ID, and device class match the approved hardware whitelist are
enumerated by the OS. VID/PID values alone are insufficient as these can be spoofed;
device class matching provides an additional layer. $\mathcal{A}_r$ cannot present
a physical USB device to a port without physical presence. Condition~(iii) is not
satisfied.

\paragraph{Channel 5.3 --- Camera.}
The camera accepts optical input requiring line-of-sight physical proximity
($r = \text{line-of-sight}$). $\mathcal{A}_r$ cannot present optical signals to a
physical camera without physical presence at the deployment location.
Condition~(iii) is not satisfied. Line-of-sight physical attacks including telephoto
optics, drone cameras, and reflection capture are $\mathcal{A}_p$ attacks addressed
by physical security controls.

\paragraph{Channel 5.4 --- Microphone.}
The microphone captures local audio requiring acoustic proximity ($r =
\text{line-of-sight}$ or contact). $\mathcal{A}_r$ cannot inject audio signals
without physical proximity. Condition~(iii) is not satisfied.

\paragraph{Channel 5.5 --- Display.}
The display is output-only. It is not an inbound channel. Conditions~(i) and~(ii)
are not applicable.

\medskip
\noindent\textbf{Conclusion:} No channel $C \in S_{\text{sovereign}}$ satisfies all
three conditions of Definition~\ref{def:remotecompromise} for $\mathcal{A}_r$.
Therefore, under the remote adversary model, remote compromise of a Sovereign AI
Architecture is not possible by any network-mediated means.
\end{proof}

\subsection{Contrast with Internet-Connected Terminal}

For an internet-connected terminal, the NIC satisfies all three conditions of
Definition~\ref{def:remotecompromise} by design:
\begin{enumerate}[label=(\roman*)]
  \item $\mathcal{A}_r$ can inject traffic on the network interface
  \item The terminal processes network traffic
  \item Network interfaces have $r = \text{remote}$
\end{enumerate}

The security question for a connected terminal is not whether a channel satisfying
conditions (i)--(iii) exists---it does, by design---but whether the network stack,
application layer, and configuration prevent exploitation of that channel. This is a
computational security question with a probabilistic answer that degrades with
adversary capability and time.

% -----------------------------------------------------------------------
\section{Lateral Movement Graph Isolation}
\label{sec:lateral}
% -----------------------------------------------------------------------

\subsection{Network Graph Model}

\begin{definition}[Institutional Network Graph]
\label{def:networkgraph}
Model the institutional network as a directed graph $G = (V, E)$ where $V$ contains
devices with network interfaces participating in the institutional network, $E
\subseteq V \times V$ is the set of directed network connections, and an adversary
controlling node $A \in V$ seeks to reach target node $T$.
\end{definition}

\begin{definition}[Lateral Movement Path]
A lateral movement path $P$ from $A$ to $T$ is a directed path $(A, v_1, v_2,
\ldots, T)$ in $G$. The existence of $P$ depends on network topology, segmentation,
and vulnerabilities present along each edge.
\end{definition}

\subsection{Graph Isolation Theorem}

\begin{theorem}[Graph Isolation]
\label{thm:graphisolation}
A terminal $T$ with no network interface is not a vertex in the institutional
network graph $G$. Consequently, no lateral movement path $P$ from any $A \in V$
to $T$ exists.
\end{theorem}

\begin{proof}
By Definition~\ref{def:networkgraph}, $V$ contains only devices with network
interfaces participating in the institutional network. A Sovereign AI Architecture
terminal has no network interface by construction. Therefore $T \notin V$. A
directed path in $G$ can only terminate at a vertex $v \in V$. Since $T \notin V$,
no path $P$ exists such that $P$ terminates at $T$.
\end{proof}

This result is independent of:
\begin{itemize}
  \item Network topology within $G$
  \item Quality of firewall and segmentation controls
  \item Sophistication of $\mathcal{A}_r$'s lateral movement capability
  \item Number of compromised nodes within $G$
\end{itemize}

Theorem~\ref{thm:graphisolation} is the strongest and most assumption-minimal
result in this paper. It requires only that the terminal has no network
interface---a physical property verifiable by inspection.

\subsection{Regulatory Implication of Graph Isolation}

Graph isolation has direct compliance implications across regulated environments:

\paragraph{Healthcare.} Under HIPAA Technical Safeguards~\cite{hipaa312}, covered
entities must implement technical controls to prevent unauthorized network access to
systems containing electronic Protected Health Information (ePHI). A terminal
outside the network graph contains no ePHI that is network-reachable, satisfying
the transmission security requirement by architectural elimination rather than
control implementation.

\paragraph{Financial Services.} PCI-DSS Requirement 1~\cite{pcidss} mandates
network security controls for systems in the cardholder data environment (CDE). A
terminal with no network interface has no network presence in the CDE and falls
outside the scope of network security controls---not by exemption, but by
architecture.

\paragraph{Defense.} CMMC Level 2 and 3~\cite{cmmc} require protection of
Controlled Unclassified Information (CUI) per NIST SP 800-171~\cite{nistSP800171}.
A terminal not on the network requires no network-layer CUI protection. Its
protection requirements reduce to physical access controls only.

\subsection{Corollary: Institutional Network Breach Does Not Propagate}

\begin{corollary}
\label{cor:breach}
A breach of the institutional network $G$ does not propagate to a Sovereign AI
Architecture terminal.
\end{corollary}

\begin{proof}
Propagation requires a lateral movement path to $T$. By
Theorem~\ref{thm:graphisolation}, no such path exists. Breach of any node or set
of nodes in $G$ does not change the fact that $T \notin V$.
\end{proof}

\subsection{Corollary: Supply Chain Network Attack Resistance}

\begin{corollary}
\label{cor:supplychain}
A compromised vendor network node cannot reach a Sovereign AI Architecture terminal
via lateral movement.
\end{corollary}

\begin{proof}
Vendor network nodes, if compromised, become members of $A \in V$. By
Theorem~\ref{thm:graphisolation}, no path from any $A \in V$ to $T$ exists when
$T \notin V$. Vendor network compromise provides no advantage toward terminal
compromise via network means.
\end{proof}

Supply chain risk is therefore reduced to physical supply chain integrity---hardware,
firmware, and manufacturing---rather than extending to operational network exposure.

\subsection{Worm Propagation Impossibility}

\begin{theorem}[Worm Propagation Impossibility]
\label{thm:worm}
Let $W$ be a self-propagating network worm operating over the institutional network
graph $G = (V, E)$. A worm propagates by traversing directed edges $(u, v) \in E$
where $v$ processes network input received from $u$.

If a terminal $T$ has no network interface, then $T \notin V$ by
Definition~\ref{def:networkgraph}. Because worm propagation requires traversal of
edges within $G$, and because $T$ is not a vertex in $G$, no sequence of
propagation steps can reach $T$. Therefore, a network worm operating within $G$
cannot infect $T$.
\end{theorem}

This result applies to all known classes of network worm propagation, including
scanning worms (WannaCry~\cite{wannacry}), topological worms
(Slammer~\cite{slammer}), and supply-chain-seeded worms (NotPetya~\cite{notpetya}).
In each case, propagation requires edge traversal within $G$. Terminals outside $G$
are unreachable by construction.

\subsection{Corollary: Botnet Enrollment Impossibility}

\begin{corollary}[Botnet Enrollment Impossibility]
\label{cor:botnet}
Botnet enrollment requires the compromised device to establish a
command-and-control (C2) communication channel with attacker infrastructure. A
Sovereign AI Architecture terminal has no network interface and therefore cannot
establish any outbound network communication. Consequently, it cannot join a botnet,
receive C2 instructions, or participate in coordinated attacker-controlled network
activity of any kind.
\end{corollary}

The practical significance of Theorem~\ref{thm:graphisolation}, Theorem~\ref{thm:worm},
and Corollaries~\ref{cor:breach}--\ref{cor:botnet} is that the four dominant
mechanisms of large-scale cyber incident propagation---lateral movement, network
breach propagation, worm infection, and botnet enrollment---are all impossible
against a Sovereign AI Architecture terminal. This follows from a single structural
property: $T \notin V$.

% -----------------------------------------------------------------------
\section{Cryptographic Session and Key Delivery Architecture}
\label{sec:crypto}
% -----------------------------------------------------------------------

\subsection{Session Key Delivery via Out-of-Band Optical Channel}

Each session employs an ephemeral asymmetric keypair $(K_{\text{pub}},
K_{\text{priv}})$ generated fresh per session by a trusted backend system. The
encrypted payload is delivered via the inbound broadcast channel encrypted under
$K_{\text{pub}}$. $K_{\text{priv}}$ is delivered to the terminal via an out-of-band
optical channel---a time-bounded optical signal presented by an authorized companion
device---and held exclusively in volatile RAM for the duration of decryption before
immediate purge.

The companion device is a patient-facing application running on the patient's
personal mobile device. It is a separate networked device outside the terminal's
trust boundary. The terminal emits nothing to it---optical signals flow from
companion to terminal for key delivery, and from terminal to companion for session
output.

\begin{definition}[Out-of-Band Optical Channel]
An out-of-band optical channel is a communication channel that:
\begin{enumerate}[label=(\roman*)]
  \item Operates via optical signal (visible or near-visible light)
  \item Requires line-of-sight physical proximity between sender and receiver
  \item Has no network reachability---$\mathcal{A}_r$ cannot inject or intercept
    content without physical presence
\end{enumerate}
\end{definition}

\subsection{Computational Security of Session Cryptography}

The session cryptography provides computational security under standard hardness
assumptions.

\begin{itemize}
  \item \textbf{Key pair:} ECC NIST P-256 or Curve25519\\
    \emph{Classical security level:} approximately 128 bits (Pollard's rho
    bound~\cite{pollard1978}---the discrete logarithm problem on these curves admits
    a square-root-time attack; 256-bit would be an overclaim)
  \item \textbf{Payload encryption:} AES-256-GCM\\
    (128-bit security under Grover's algorithm~\cite{grover1996} against a quantum
    adversary)
  \item \textbf{Authentication:} HMAC-SHA256 over canonical request body
\end{itemize}

We note explicitly that 128-bit classical security---not 256-bit---is the correct
characterization for P-256 and Curve25519~\cite{bernstein2014,pollard1978}. Against
a cryptographically relevant quantum adversary employing Shor's
algorithm~\cite{shor1997}, ECC provides no security guarantee. Migration to
post-quantum key encapsulation (ML-KEM per FIPS 203~\cite{fips203}) is recommended
for deployments requiring long-term quantum resistance. AES-256-GCM retains 128-bit
security against Grover's algorithm and is considered quantum-resistant at current
key length.

\subsection{Physical Constraint on $K_{\text{priv}}$ Interception}

For $\mathcal{A}_r$ to obtain $K_{\text{priv}}$, the following must hold
simultaneously:
\begin{enumerate}[label=(\roman*)]
  \item Physical presence at the terminal location (line-of-sight to companion
    device)
  \item Optical capture capability within capture range
  \item Capture within the time-bounded expiry window $T_{\text{exp}}$
  \item Prior targeting knowledge of the specific session
\end{enumerate}

The conjunction of (i)--(iv) defines an $\mathcal{A}_p$ attack. This represents a
categorical threat model reduction: a global network adversary is reduced to a
locally-present, time-constrained physical adversary. Line-of-sight physical attacks
on the optical channel---including telephoto lens capture, drone cameras, and
reflection-based capture---remain within the $\mathcal{A}_p$ class and are addressed
by physical security controls at the deployment location.

A remote adversary who has previously compromised the companion device may attempt
to use it as a proxy---reading the optical signal via the device camera and
exfiltrating $K_{\text{priv}}$ over the companion device's network connection. This
attack requires prior targeted compromise of a specific patient's personal device and
physical presence of that device at the terminal, placing it firmly within the
$\mathcal{A}_p$ class; it does not constitute a remote network attack under
Definition~2.1, but it does underscore that the companion device is an attack surface
whose integrity is outside the terminal's trust boundary.

\subsection{Key Lifecycle and Volatile Memory Confinement}

$K_{\text{priv}}$ is confined to volatile RAM from optical reception to
post-decryption purge:
\[
  t_{\text{exist}}(K_{\text{priv}}) = [t_{\text{reception}},\; t_{\text{purge}}]
\]
where $t_{\text{purge}}$ follows $t_{\text{decryption}}$ by at most one memory
management cycle. $K_{\text{priv}}$ is never written to persistent storage under
any operating condition. This property prevents forensic recovery of
$K_{\text{priv}}$ after session termination.

\subsection{Outbound Data via Optical Channel}

Session output is transmitted from the terminal to the companion device via an
optical channel displayed on the terminal screen. Interception requires
line-of-sight physical presence and optical capture within the display window---an
$\mathcal{A}_p$ attack with no network path.

\subsection{Complete Out-of-Band Session Chain}

The complete session data flow involves no network-accessible data in decryptable
form at any stage:

\begin{description}
  \item[Stage 1:] Encrypted payload $\to$ broadcast/managed inbound channel
    (one-way) $\to$ terminal.\\
    \emph{[Encrypted under $K_{\text{pub}}$; computationally indistinguishable from
    random without $K_{\text{priv}}$]}
  \item[Stage 2:] $K_{\text{priv}}$ $\to$ optical channel (time-bounded
    $T_{\text{exp}}$, line-of-sight only) $\to$ terminal.\\
    \emph{[$K_{\text{priv}}$ never on any network; not reachable by $\mathcal{A}_r$]}
  \item[Stage 3:] Decryption $\to$ volatile RAM only $\to$ $K_{\text{priv}}$
    purged immediately.\\
    \emph{[No persistent storage under any operating condition]}
  \item[Stage 4:] Session output $\to$ optical channel $\to$ companion device.\\
    \emph{[Not transmitted on any network; interception requires physical presence]}
\end{description}

At no stage does decryptable sensitive data traverse a network interface.

\subsection{Cross-Domain Applicability of Session Architecture}

The out-of-band optical key delivery architecture is not specific to any single
regulated domain:

\paragraph{Healthcare.} Patient identity and clinical data never traverse a network
in decryptable form. Re-identification at the terminal is architecturally impossible
without physical presence and possession of the time-bounded optical key.

\paragraph{Financial Services.} Transaction authorization keys delivered out-of-band
cannot be intercepted by network-layer adversaries. Man-in-the-middle attacks on the
key delivery channel require physical presence.

\paragraph{Defense.} Session keys delivered via out-of-band optical channel cannot
be exfiltrated via network compromise. Physical access controls governing the
terminal location become the primary security perimeter.

% -----------------------------------------------------------------------
\section{Composite Threat Model Reduction}
\label{sec:composite}
% -----------------------------------------------------------------------

\subsection{Summary of Eliminated Attack Classes}

The Sovereign AI Architecture eliminates the following attack classes by
construction. These eliminations are structural---they do not depend on the quality
of security controls, configuration, or monitoring:

\begin{table}[ht]
\centering
\caption{Eliminated attack classes vs.\ connected terminal}
\begin{tabularx}{\textwidth}{@{}Xll@{}}
\toprule
\textbf{Attack Class} & \textbf{Connected Terminal} & \textbf{Sovereign AI Arch.} \\
\midrule
Remote code execution & Possible via NIC & Eliminated ($S_r = \emptyset$) \\
Network scanning & Terminal visible & Eliminated ($T \notin V$) \\
Lateral movement & Possible via $G$ & Eliminated ($T \notin V$) \\
Worm propagation & Possible & Eliminated (Thm.~\ref{thm:worm}) \\
Botnet enrollment & Possible & Eliminated (Cor.~\ref{cor:botnet}) \\
Man-in-the-middle & Possible & Eliminated (no path) \\
Remote cred.\ theft & Possible & Eliminated (no svc) \\
C2 callback & Possible & Eliminated (no out) \\
Network-based APT & Possible & Eliminated ($S_r = \emptyset$) \\
Vendor network pivot & Possible & Eliminated (Cor.~\ref{cor:supplychain}) \\
\bottomrule
\end{tabularx}
\end{table}

\begin{table}[ht]
\centering
\caption{Residual attack classes requiring $\mathcal{A}_p$}
\begin{tabularx}{\textwidth}{@{}Xl@{}}
\toprule
\textbf{Attack Class} & \textbf{Mitigation} \\
\midrule
Optical $K_{\text{priv}}$ capture & Time-bounded $T_{\text{exp}}$, physical security \\
Optical output capture & Physical security, display orientation \\
USB device injection & VID/PID/class whitelist + attestation \\
Chassis intrusion & Tamper detection, volatile wipe \\
RF injection near antenna & Physical security perimeter \\
Broadcast injection & Signature verification (Assumption~\ref{asm:auth}) \\
\bottomrule
\end{tabularx}
\end{table}

\subsection{Formal Threat Model Reduction Statement}

\begin{theorem}[Threat Model Reduction]
\label{thm:reduction}
Under Assumptions~\ref{asm:hardware}--\ref{asm:physical}, a Sovereign AI
Architecture with a receiver-side unidirectional inbound channel satisfying
Definition~\ref{def:unidirectionality} reduces the adversary threat model from
$\{\mathcal{A}_r, \mathcal{A}_p\}$ to $\{\mathcal{A}_p\}$ for all
network-mediated attack classes.
\end{theorem}

\begin{proof}
By Theorem~\ref{thm:surfaceelim}, $S_{\text{sovereign},r} = \emptyset$. By
Theorem~\ref{thm:remoteimpossibility}, no channel satisfying the remote compromise
condition exists for $\mathcal{A}_r$. By Theorem~\ref{thm:graphisolation}, no
lateral movement path to $T$ exists in $G$. By Theorem~\ref{thm:worm}, worm
propagation cannot reach $T$. By Corollary~\ref{cor:botnet}, botnet enrollment is
impossible. By Section~\ref{sec:crypto}, all session data flows are confined to
optical channels requiring physical presence. Therefore all attack classes requiring
only $\mathcal{A}_r$ capabilities are eliminated. The residual attack surface
requires $\mathcal{A}_p$ capabilities.
\end{proof}

\[
  A_{\text{elim}} = \{\text{all network-mediated attack classes}\}
\]

The eliminated set constitutes the dominant attack surface for connected terminals
in regulated environments~\cite{hipaajournal2024,ponemon2023,verizonDBIR2023,cisaAPT2022,wannacry,notpetya}.

\subsection{Residual Physical Attack Surface}
\label{sec:residual}

The residual attack surface requires $\mathcal{A}_p$. This is the irreducible
minimum for any deployed physical system:
\begin{align*}
  R_{\text{sovereign}} &= R_{\text{physical}} \quad (\text{irreducible minimum})\\
  R_{\text{connected}} &= R_{\text{physical}} + R_{\text{network}}\\
  \Delta R &= R_{\text{network}} > 0
\end{align*}

\subsection{Limitations}
\label{sec:limitations}

\paragraph{Limitation 8.1 --- Hardware Trust.}
Assumption~\ref{asm:hardware} bounds the proof. A compromised secure enclave, BIOS,
or hardware accelerator introduced during manufacturing could undermine key storage
guarantees. This is the architecture's most load-bearing unexamined assumption: the
entire threat model reduction holds only if the hardware root of trust was not
subverted before deployment. The trust chain terminates at a point that the
architecture itself cannot verify at runtime.

Candidate mitigations operate at manufacture and provisioning time rather than at
runtime. TPM 2.0 remote attestation~\cite{tpm2} allows a verifier to confirm that
the device is running known-good firmware and that key material was generated within
a genuine enclave. ARM TrustZone measured boot~\cite{armtrustzone} extends this to
the full boot chain. NIST SP 800-193 provides a framework for detection and recovery
from firmware compromise. Collectively these mechanisms reduce supply chain risk but
do not eliminate it. Full treatment of hardware supply chain integrity is outside the
scope of this paper.

\paragraph{Limitation 8.2 --- Broadcast Transmitter Trust.}
The broadcast transmitter operator is a trusted party. Compromise of the signing key
would allow malicious firmware injection. This is mitigated by Assumption~\ref{asm:auth}
and key management practices outside the scope of this paper.

\paragraph{Limitation 8.3 --- Implementation Correctness.}
This paper proves architectural properties, not implementation correctness. Software
vulnerabilities in the broadcast decoder, key management layer, or inference engine
are addressed by secure development lifecycle practices~\cite{nistSSDF,iec62304}.

\paragraph{Limitation 8.4 --- Quantum Adversary.}
ECC provides no security guarantee against a cryptographically relevant quantum
adversary. AES-256-GCM retains 128-bit security under Grover's algorithm. Migration
to ML-KEM (FIPS 203~\cite{fips203}) is recommended for long-term quantum
resistance.

\paragraph{Limitation 8.5 --- USB Whitelisting.}
VID/PID matching alone is insufficient against a determined physical adversary with
device-spoofing capability. VID/PID plus device class matching provides additional
resistance; cryptographic device attestation should be considered for high-assurance
deployments.

\paragraph{Limitation 8.6 --- Revocation Latency.}
The absence of a return channel eliminates pull-based revocation mechanisms. A
terminal cannot query a Certificate Revocation List (CRL) or receive an emergency
kill signal on demand. Revocation of a compromised terminal Device ID is
push-based: a signed revocation payload must be queued and delivered via the
broadcast channel at the next available broadcast cycle. The latency between
compromise detection and confirmed revocation delivery is therefore bounded by
broadcast cycle frequency and operator response time. For deployments where rapid
revocation is operationally critical, broadcast cycle frequency should be sized
accordingly, and the revocation latency window should be treated as a residual risk
alongside the physical attack surface.

% -----------------------------------------------------------------------
\section{Discussion}
\label{sec:discussion}
% -----------------------------------------------------------------------

\subsection{Prototype Implementation}

A reference implementation of the described architecture has been developed and
validated across the following components:

\paragraph{Inbound Data Channel.} Satellite broadcast (DVB-S2) receive-only
terminal. Addressed IP payloads delivered via GSE stream addressing to a specific
terminal Device ID. No return channel present in hardware or software configuration.
Receive-only VSAT downlink hardware contains no transmit circuitry on the downlink
frequency.

\paragraph{On-Device AI Inference.} Quantized large language model (2B parameter
class, INT4 quantization) executing via \texttt{llama.cpp} on a dedicated neural
processing unit (10+ TOPS). Full inference pipeline runs locally with no cloud
dependency. Speech transcription via \texttt{whisper.cpp}. Computer vision pipeline
via MediaPipe. Inference latency target: sub-3 seconds per session on target SoC.

\paragraph{Out-of-Band Key Exchange.} Ephemeral asymmetric keypair (ECC P-256)
generated per session. $K_{\text{priv}}$ delivered via optical channel with
time-bounded TTL. $K_{\text{priv}}$ confined to volatile RAM, purged immediately
post-decryption. Session output delivered via optical channel.

\paragraph{USB Hardware Whitelisting.} Kernel-level eBPF hook enforcing
VID/PID/device-class whitelist. Only approved peripherals enumerated by OS.

\paragraph{Secure Enclave.} Hardware-bound root key burned at manufacture. Key
derivation executing within enclave. Key material never written to persistent
storage.

The prototype demonstrates the architectural claims of this paper in a functional
implementation: no network interface is present or required during operation, all
data flows are out-of-band optical or receive-only broadcast, and full AI inference
executes on-device. Empirical characterization of latency distributions, throughput
under load, and failure modes under production conditions is left to future work.

A consequential property of on-device inference is graceful degradation under
broadcast unavailability. The terminal does not require a live connection to the deployment
backend to conduct a session. The clinical model, inference pipeline, and firmware
are resident on the device and persist across broadcast cycles. What does require a
fresh broadcast per session is the encrypted session payload---the per-patient data
addressed to a specific terminal Device ID. If the broadcast is unavailable at
session time, the terminal retains full inference capability but cannot receive new
session payloads until the broadcast resumes.

\subsection{Broadcast Infrastructure for Sovereign AI Deployment}

Any broadcast medium satisfying Definition~\ref{def:unidirectionality} serves as
sovereign inbound infrastructure. A single satellite footprint covers entire
continents; a terrestrial tower covers a metropolitan area containing potentially
thousands of deployed terminals. The broadcast channel simultaneously delivers
addressed payloads to individual terminals and fleet-wide firmware updates with zero
network infrastructure dependency at the receiving end.

Adding terminals to a sovereign deployment adds zero network complexity at the
receiving end. There are no firewall rules to configure, no VLANs to provision, no
IT tickets to file at the terminal site. The terminal requires power and an antenna
feed. Backend operations---Device ID provisioning, GSE stream addressing, and session
key issuance---are one-time per-terminal setup steps handled centrally. The broadcast
layer handles distribution.

A practical constraint on model update delivery is broadcast throughput relative to
model size. A 2B parameter INT4-quantized model ($\sim$1 GB) requires approximately
2.5 minutes at typical satellite or terrestrial broadcast data rates. A 70B parameter
model at equivalent quantization ($\sim$35 GB) requires substantially more channel
time. For large model classes, delta-updates and quantization-aware patching are the
standard operational path.

Centralized virtual terminals minimize operational maintenance cost but introduce
systemic risk: compromise of the central infrastructure propagates to the entire
deployed fleet simultaneously. Real-world incidents including the 2024 Change Healthcare ransomware attack~\cite{hhs2024change},
the Ascension Health network shutdown~\cite{ascension2024}, and the CommonSpirit
Health incident~\cite{commonspirit2022} demonstrate that this is not a theoretical concern. The
Sovereign Broadcast Architecture makes the deliberate opposite trade-off: higher
per-terminal autonomy in exchange for elimination of fleet-wide compromise pathways.

\subsection{Broadcast Zero-Trust Edge Computing}

We propose \emph{Broadcast Zero-Trust Edge Computing} as a generalizable
architecture pattern defined by four properties:

\begin{enumerate}
  \item \textbf{Sovereign Compute:} All AI inference executes on-device. No cloud
    dependency exists for any inference operation.
  \item \textbf{Unidirectional Inbound Channel:} All inbound data arrives via a
    channel satisfying Definition~\ref{def:unidirectionality}. Enforcement is
    hardware-based in broadcast and satellite configurations, policy-based in managed
    IP configurations.
  \item \textbf{Optical Identity:} Session key material and output data are exchanged
    via out-of-band optical channels requiring physical presence. No network path
    exists for identity or session data.
  \item \textbf{Network Attack Surface Elimination:} The terminal has no network
    interface. $S_r = \emptyset$. $T \notin V$.
\end{enumerate}

This pattern extends the Zero Trust principle~\cite{nistZTA,kindervag2010} to its
logical limit. Zero Trust assumes network presence and works backward to verify every
transaction~\cite{ward2014}. Broadcast Zero-Trust eliminates network presence as a
design primitive---there are no transactions to verify because there is no network.
The security guarantee is not stronger policy. It is the absence of the policy
surface entirely.

\begin{quote}
\emph{Design Principle --- Isolation by Elimination of Return Channels:}
A central design principle emerging from this architecture is that the most reliable
way to eliminate remote attack surface is not to secure bidirectional communication
channels, but to eliminate them entirely. By restricting external input to
authenticated broadcast and removing all return paths from the terminal to external
networks, the system converts remote compromise from a probabilistic security problem
into a structural impossibility under the defined threat model. Operational
observability and fleet management are then derived from side-effects of normal
operation rather than from dedicated management channels.

In other words: \emph{the safest remote interface is no remote interface.}
\end{quote}

\subsection{Relationship to Prior High-Assurance System Work}

Data diode architectures~\cite{tuptuk2018,waterfall2022} enforce unidirectionality
at the hardware level for critical infrastructure. Our contribution extends this to
AI inference terminals and formalizes receiver-side unidirectionality as a deployment
configuration property across broadcast, satellite, and managed IP configurations.

The seL4 verified microkernel~\cite{sel4} provides formal proof of OS isolation
properties at the software level. Our argument operates at the network architecture
level and is complementary---seL4-class verification could strengthen
Limitation~8.3. Intel SGX~\cite{intelsgx} and Google Titan~\cite{titansecurity}
provide hardware root of trust for computation, directly relevant to
Assumption~\ref{asm:hardware}. NSA CNSA 2.0~\cite{cnsa2} provides algorithm
selection guidance for national security systems including post-quantum migration
timelines. The Howard/Wing attack surface metric~\cite{howard2003,manadhata2011}
provides the structured attack surface framework adopted in
Section~\ref{sec:attacksurface}. BeyondCorp~\cite{ward2014} represents the most
prominent industrial deployment of Zero Trust principles.

\subsection{Cross-Domain Applicability}

The threat model reduction argument is domain-agnostic. The properties that make a
Sovereign AI Architecture secure in a clinical environment are identical to those
that make it secure on a trading floor or in a classified facility:

\begin{itemize}
  \item No network interface $\Rightarrow T \notin V \Rightarrow$ no lateral
    movement path
  \item Receive-only broadcast channel $\Rightarrow S_r = \emptyset \Rightarrow$ no
    remote attack vector
  \item Out-of-band optical key delivery $\Rightarrow$ physical presence required
    for any session interception
  \item $T \notin V \Rightarrow$ worm propagation and botnet enrollment impossible
\end{itemize}

The regulatory compliance implications differ by domain
(Section~\ref{sec:lateral}.3) but the underlying security argument is identical.

\subsection{Fleet Health Monitoring via Session Lifecycle Events}

The architecture requires no dedicated health monitoring channel. Terminal liveness
and session completion are observable as a natural side-effect of normal operation.
When a user initiates a session, the companion device scans the session-start optical
code, and the terminal Device ID and precise timestamp are forwarded to the
deployment backend. When the session concludes, the companion device scans the
session-end code and forwards the same Device ID, timestamp, and session duration.
The backend receives two signals per session---session-start and session-end---with
no modification to the terminal, no new QR payload, and no new channel.

Each terminal exhibits a characteristic session cadence reflecting the usage profile
of its deployment environment. The fleet backend learns this pattern per terminal and
flags anomalies statistically. No operator-defined thresholds are required.

\paragraph{Operational Response Model.}

When anomalous inactivity is detected, the fleet backend notifies the operator with
the terminal Device ID, deployment location, last known session timestamp, and
inferred anomaly type. An authorized operator selects from three response tiers
ordered by escalation cost.

\begin{description}
  \item[Tier 1 --- Remote Restart and Reprovision.] The operator issues restart or
    firmware update commands from the Customer Service App without leaving the
    operations center. Commands are queued and delivered via the authenticated
    broadcast channel at the next available broadcast cycle. Delivery confirmation
    is indicated by a subsequent session-start signal.

  \item[Tier 2 --- In-Person Troubleshoot.] The operator or facility staff visits
 the terminal with the Customer Service App and issues a troubleshoot command
 through the authenticated broadcast channel. The terminal executes a diagnostic
    routine and displays the results---firmware version, model version, peripheral
    status, error codes, broadcast signal strength, storage health, and last known
    operational state---on its screen. Physical presence is required to read the
    diagnostic display. The results cannot be transmitted remotely.

  \item[Tier 3 --- Onsite Service Visit.] A field technician is dispatched with a
    pre-imaged replacement unit. The failed unit is swapped and returned for forensic
    analysis. Tier 3 is the irreducible minimum for hardware failure in any deployed
    physical system.
\end{description}

\paragraph{Limitations of Session-Lifecycle Telemetry.}

This model has inherent constraints. First, observability is user-gated: a terminal
with no active users generates no session signals. Second, the model detects absence
but not cause. Third, delivery of restart, firmware, and troubleshoot commands via
broadcast is not acknowledged. Fourth, the companion device is a dependency in the
relay chain---a lost or malfunctioning companion device will suppress session signals
from an otherwise healthy terminal.

No remote access to the terminal exists or is required at any point in this
procedure. Restart, firmware, and troubleshoot commands all transit the same
authenticated broadcast inbound channel already present in the architecture.
Operational visibility is a consequence of the session lifecycle, not an addition
to it.

% -----------------------------------------------------------------------
\section{Conclusion}
\label{sec:conclusion}
% -----------------------------------------------------------------------

We have presented a formal argument under a defined adversary model establishing
that a Sovereign AI Architecture in which inbound data is delivered via a physically 
unidirectional channel—implemented using receive-only broadcast infrastructure or 
certified hardware data diodes—with no return path to any external network
satisfying Definition~\ref{def:unidirectionality} eliminates all network-mediated
compromise vectors. The argument is established across four complementary
frameworks---attack surface decomposition, channel reachability analysis, graph
isolation, and cryptographic session architecture---each independently sufficient and
mutually reinforcing.

The core results are:

\begin{enumerate}
  \item \textbf{Remote attack surface elimination (Theorem~\ref{thm:surfaceelim}):}
    $S_{\text{sovereign},r} = \emptyset$---no interface reachable by $\mathcal{A}_r$
    exists in a Sovereign AI Architecture.

  \item \textbf{Remote compromise impossibility (Theorem~\ref{thm:remoteimpossibility}):}
    Under the remote adversary model and authenticated inbound channel assumption,
    no channel satisfying the remote compromise condition exists for $\mathcal{A}_r$.

  \item \textbf{Lateral movement impossibility (Theorem~\ref{thm:graphisolation}):}
    $T \notin V$---a terminal with no network interface has no vertex in the
    institutional network graph and cannot be reached by any lateral movement path.
    This is the strongest and most assumption-minimal result in the paper.

  \item \textbf{Worm propagation impossibility (Theorem~\ref{thm:worm}):}
    Self-propagating network malware cannot reach a terminal that is not a vertex in
    the institutional network graph. This applies to all known propagation
    mechanisms including scanning, topological, and supply-chain-seeded worms.

  \item \textbf{Botnet enrollment impossibility (Corollary~\ref{cor:botnet}):}
    A terminal with no outbound network interface cannot establish C2 communication
    and cannot be enrolled in a botnet under any circumstance.

  \item \textbf{Physical threat model reduction (Theorem~\ref{thm:reduction}):}
    The adversary threat model is reduced from $\{\mathcal{A}_r, \mathcal{A}_p\}$
    to $\{\mathcal{A}_p\}$---from a global network adversary to a locally-present
    physical adversary---for all network-mediated attack classes.

  \item \textbf{Receiver-side unidirectionality (Theorem~\ref{thm:returnchannel}):}
    Under hardware-enforced configurations, the inbound channel cannot transmit by
    any software means on the receiver device. The security argument holds under
    either enforcement tier.

  \item \textbf{Broadcast Zero-Trust Edge Computing:} A generalizable architecture
    pattern extending the Zero Trust principle to its logical limit---eliminating
    network presence rather than verifying it.A corollary design principle: the most 
    reliable way to eliminate remote attack surface is not to secure bidirectional 
    channels but to eliminate them entirely, converting remote compromise from a 
    probabilistic problem into a structural impossibility.

  \item \textbf{Fleet observability without new channels:} Terminal health and
    operational status are fully inferable from session lifecycle signals already
    present in the architecture. The attack surface is unchanged.
\end{enumerate}

We have been careful to distinguish architectural guarantees from implementation
guarantees, computational security from information-theoretic security, and structural
properties from probabilistic claims. The argument holds under clearly stated
assumptions, and its limitations are explicitly identified.

The practical conclusion for regulated environment deployment is precise: an
institution deploying a Sovereign AI terminal does not need to trust its own network
security posture to protect the terminal. The terminal is outside the network graph.
Its security does not degrade when the institutional network is breached---because
the institutional network has no path to it.

The four dominant mechanisms of large-scale cyber incident propagation---lateral
movement, network breach propagation, worm infection, and botnet enrollment---are all
impossible against this architecture. They follow from a single structural property:
$T \notin V$.

The architecture does not harden against network attacks. It eliminates the network
attack class.

% -----------------------------------------------------------------------
% Bibliography
% -----------------------------------------------------------------------

\appendix

\section{Deployment Configuration Reference Implementations}
\label{app:configurations}

This appendix describes three concrete deployment configurations satisfying Definition~4.1 (receiver-side unidirectionality). The formal security argument of Sections~3--8 is independent of which configuration is chosen; configurations differ in how strongly they enforce unidirectionality and in their operational characteristics. Table~\ref{tab:config-comparison} summarizes all three.

\subsection{Configuration 1: Physically Unidirectional Broadcast Terminal}

\paragraph{Overview.}
A receive-only broadcast or satellite terminal receives data via licensed radio-frequency spectrum on a downlink frequency for which the device contains no transmit circuitry. Unidirectionality is enforced by the physical absence of transmit-capable components, not by software policy or configuration. This is the highest-assurance hardware configuration.

\paragraph{Reference implementation.}
The prototype described in Section~9.1 uses a DVB-S2 receive-only VSAT downlink terminal. Addressed IP payloads are delivered via Generic Stream Encapsulation (GSE) over a geostationary satellite link. The terminal's Device ID is embedded in the GSE stream label; the terminal processes only streams addressed to its Device ID or the broadcast group address.

\begin{itemize}
  \item \textbf{Hardware:} DVB-S2/S2X receive-only VSAT downlink receiver
  \item \textbf{Encapsulation:} GSE per ETSI EN 302 307
  \item \textbf{Addressing:} per-terminal Device ID in GSE stream label; broadcast group address for fleet-wide payloads
  \item \textbf{Transmit circuitry:} absent on downlink frequency---hardware enforcement
  \item \textbf{Return channel:} none
  \item \textbf{Equivalent standards:} DVB-T2 (terrestrial, Europe/Asia/Africa), ATSC 3.0 / NextGen TV (terrestrial, North America/South Korea) [5, 6, 17], ISDB-T (Japan/South America), DAB+ and HD Radio (radio datacasting)
\end{itemize}

\paragraph{Enforcement tier.}
Hardware-enforced. Software on the receiver device cannot provision transmit capability on the downlink frequency because no transmit hardware exists at that frequency. Satisfies Definition~4.1 unconditionally. Corresponds to the \emph{Yes} rows of Table~1.

\begin{table}[h]
\centering
\begin{tabular}{ll}
\hline
\textbf{Property} & \textbf{Value} \\
\hline
Unidirectionality enforcement & Hardware (transmit circuitry absent) \\
Definition 4.1 satisfied & Yes---unconditionally \\
Typical downlink throughput & 20--100\,Mbps (DVB-S2X ACM) \\
Coverage footprint & Continental (GEO) or global (LEO) \\
Terminal count scalability & Unlimited---broadcast is one-to-many \\
Network infrastructure at receiver & None---power and antenna feed only \\
Return channel for fleet management & None---push-only via broadcast \\
Model update delivery (${\sim}1$\,GB) & ${\sim}2.5$\,min at 50\,Mbps \\
\hline
\end{tabular}
\caption{Configuration 1 operational characteristics.}
\label{tab:config1}
\end{table}

\subsection{Configuration 2: Hardware Data Diode}

\paragraph{Overview.}
A hardware data diode is a certified unidirectional security gateway enforcing one-way data flow at the hardware level between a source network and a destination enclave. Unlike broadcast configurations, the data diode operates on a point-to-point or hub-to-terminal managed network segment. Unidirectionality is enforced by certified hardware and is independent of software correctness on the terminal.

\paragraph{Reference architecture.}
A backend data feed is delivered from a source enclave through a certified hardware data diode to the terminal's inbound interface. The terminal has no network interface other than the diode receive port; the diode's hardware design prevents any signal from propagating from the terminal back toward the source network.

\begin{itemize}
  \item \textbf{Gateway device:} certified hardware data diode (e.g., Owl Cyber Defense, Waterfall Security)
  \item \textbf{Certification:} Common Criteria EAL4+ or equivalent national scheme
  \item \textbf{Topology:} source enclave $\rightarrow$ [data diode] $\rightarrow$ terminal receive port
  \item \textbf{Payload authentication:} cryptographic signature verification required (Assumption~2.3)
  \item \textbf{Return channel:} physically impossible toward source network; out-of-band optical channel for session output (Section~7.5)
\end{itemize}

\paragraph{Enforcement tier.}
Hardware-enforced at the diode boundary, independent of software on either the source or receiving system. Satisfies Definition~4.1 unconditionally. Corresponds to the \emph{Yes (hardware data diode)} row of Table~1.

\begin{table}[h]
\centering
\begin{tabular}{ll}
\hline
\textbf{Property} & \textbf{Value} \\
\hline
Unidirectionality enforcement & Hardware (certified data diode) \\
Definition 4.1 satisfied & Yes---unconditionally \\
Typical throughput & Up to 10\,Gbps (device-dependent) \\
Coverage & Point-to-point or hub-and-spoke \\
Network infrastructure at receiver & Diode receive port only---no bidirectional NIC \\
Return channel for fleet management & None toward source network \\
Primary deployment context & Critical infrastructure, classified enclaves \\
\hline
\end{tabular}
\caption{Configuration 2 operational characteristics.}
\label{tab:config2}
\end{table}

\subsection{Configuration 3: Managed IP with Outbound Disabled}

\paragraph{Overview.}
The terminal connects to a managed IP network segment on which outbound traffic is disabled by kernel-level and perimeter firewall policy. The terminal retains a bidirectional-capable network interface, but the outbound path is blocked at the OS level (e.g., \texttt{iptables} DROP on the OUTPUT chain) and enforced at the network perimeter. This is a policy-enforced configuration: the return channel is absent under correct operation, but the hardware retains transmit capability.

\begin{quote}
\textit{Security note.} This configuration provides weaker unidirectionality assurance than Configurations~1 and~2 because enforcement depends on the correctness of kernel policy and perimeter controls. A software vulnerability, misconfiguration, or privileged process on the terminal could in principle re-enable outbound capability. Deployments using this configuration should treat policy correctness as an additional assumption alongside Assumptions~2.1--2.4.
\end{quote}

\paragraph{Reference implementation.}
The terminal connects to a dedicated VLAN with no default gateway toward the public internet. The upstream firewall enforces an egress ACL dropping all traffic sourced from the terminal's IP address. On the terminal, \texttt{iptables}/\texttt{nftables} drops all OUTPUT traffic except loopback. Inbound payload delivery uses standard IP unicast or multicast from an authenticated backend.

\begin{itemize}
  \item \textbf{Network interface:} standard Ethernet NIC (bidirectional hardware)
  \item \textbf{Kernel policy:} OUTPUT chain DROP for all non-loopback destinations
  \item \textbf{Perimeter enforcement:} upstream firewall egress ACL blocking terminal source IP
  \item \textbf{Return channel:} absent under correct operation; hardware retains transmit capability
  \item \textbf{Additional assumption required:} outbound blocking policy correctly implemented and maintained
\end{itemize}

\paragraph{Enforcement tier.}
Policy-enforced (Conditional in Table~1). The security argument of Sections~5--6 holds under this configuration subject to correct implementation and operation of the outbound blocking policy.

\begin{table}[h]
\centering
\begin{tabular}{ll}
\hline
\textbf{Property} & \textbf{Value} \\
\hline
Unidirectionality enforcement & Kernel policy + perimeter firewall \\
Definition 4.1 satisfied & Conditional---absent under correct operation \\
Typical throughput & Standard Ethernet/IP rates \\
Coverage & Institutional LAN/WAN \\
Network infrastructure at receiver & Standard NIC---full hardware retained \\
Return channel if policy fails & Possible \\
Primary deployment context & Environments without broadcast/diode infrastructure \\
\hline
\end{tabular}
\caption{Configuration 3 operational characteristics.}
\label{tab:config3}
\end{table}

\subsection{Configuration Comparison}

\begin{table}[h]
\centering
\small
\begin{tabular}{lccc}
\hline
\textbf{Property} & \textbf{Config 1: Broadcast} & \textbf{Config 2: Diode} & \textbf{Config 3: Managed IP} \\
\hline
Def.\ 4.1 satisfied        & Yes (unconditional)   & Yes (unconditional)    & Conditional \\
Unidirectionality basis    & No TX circuitry       & Certified HW diode     & Kernel + firewall policy \\
SW subversion possible     & No                    & No (at diode boundary) & Yes (if policy fails) \\
Coverage model             & Broadcast 1:many      & Point-to-point         & IP 1:1 or 1:many \\
Return channel possible    & No                    & No                     & Yes if misconfigured \\
Fleet scalability          & Unlimited             & One port per terminal  & Standard IP scaling \\
Infrastructure at terminal & Power + antenna       & Diode receive port     & Standard NIC \\
\hline
\end{tabular}
\caption{Comparison of deployment configurations against key security properties.}
\label{tab:config-comparison}
\end{table}

\textbf{Note A.1 (5G Broadcast).}
3GPP Release~17 defines NR Multicast and Broadcast Services (MBS), enabling 5G networks to operate in a broadcast-only device profile in which the UE does not establish a bidirectional connection. In this profile the device satisfies Definition~4.1 conditionally---the platform hardware retains bidirectional capability but the device profile suppresses uplink. This is analogous to Configuration~3 and is included in Table~1 as \emph{Conditional}. Operators should verify device profile enforcement at the hardware modem level for high-assurance deployments.

\end{document}